\documentclass[cmp]{svjour}  
\usepackage{amsmath}
\usepackage{amsfonts,amssymb,bm,hyperref}
\numberwithin{equation}{section}

\journalname{Communications in Mathematical Physics}

\begin{document}

\title{{Modeling the electron with Cosserat elasticity}}
\titlerunning{Modeling the electron with Cosserat elasticity}

\author{James Burnett\inst{1} and Dmitri Vassiliev\inst{2}}
\institute{
Department of Mathematics and Department of Physics \& Astronomy,
University College London,
Gower Street,
London WC1E 6BT,
UK.
\\ \email{J.Burnett@ucl.ac.uk}
\and
Department of Mathematics and Institute of Origins,
University College London,
Gower Street,
London WC1E 6BT,
UK.
\\ \email{D.Vassiliev@ucl.ac.uk}
}
\authorrunning{J. Burnett and D. Vassiliev}

\maketitle
\begin{abstract}
We suggest an alternative mathematical model for the electron in
dimension 1+2. We think of our (1+2)-dimensional spacetime as an
elastic continuum whose material points can experience no
displacements, only rotations. This framework is a special case of
the Cosserat theory of elasticity. Rotations of material points are
described mathematically by attaching to each geometric point an
orthonormal basis which gives a field of orthonormal bases called
the coframe. As the dynamical variables (unknowns) of our theory we
choose a coframe and a density. We then add an extra (third) spatial
dimension, extend our coframe and density into dimension 1+3, choose
a conformally invariant Lagrangian proportional to axial torsion
squared, roll up the extra dimension into a circle so as to
incorporate mass and return to our original (1+2)-dimensional
spacetime by separating out the extra coordinate. The main result of
our paper is the theorem stating that our model is equivalent to the
Dirac equation in dimension 1+2. In the process of analyzing our
model we also establish an abstract result, identifying a class of
nonlinear second order partial differential equations which reduce
to pairs of linear first order equations.
\end{abstract}

\section{Introduction}
\label{Introduction}

In this paper we consider an electron living in (1+2)-dimensional
Minkowski spacetime $\mathbb{M}^{1+2}$ with coordinates $x^\alpha$,
$\alpha=0,1,2$, and metric
$g_{\alpha\beta}=\operatorname{diag}(-1,+1,+1)$.
The reduction of spatial
dimension from 3 to 2 makes dealing with spin easier as we have only
two possibilities, spin up and spin down. At a technical level this
reduction of spatial dimension manifests itself in the fact that we
do not need a 4-component complex bispinor for describing the
electron, just a 2-component complex spinor.

The Dirac equation in $\mathbb{M}^{1+2}$ is
\begin{equation}
\label{Dirac equation}
[\sigma^\alpha{}_{\dot ab}(i\partial+A)_\alpha
\pm m\sigma^3{}_{\dot ab}]\eta^b=0.
\end{equation}
Here $m$ is the electron mass, $\sigma^\alpha$ are Pauli matrices
(see (\ref{Pauli matrix temporal}), (\ref{Pauli matrices spatial})),
$\partial_\alpha=\partial/\partial x^\alpha$
and $A_\alpha$ is a given external real electromagnetic field.
The tensor summation index $\alpha$ runs through the values $0,1,2$, the
spinor summation index $b$ runs through the values $1,2$
and the free spinor index $\dot a$ runs through the values $\dot1,\dot2$.
The spinor field $\eta:\mathbb{M}^{1+2}\to\mathbb{C}^2$ is the dynamical
variable (unknown quantity).
The two choices of sign give two versions of the Dirac equation
corresponding to spin up and down.

Equations (\ref{Dirac equation}) are, of course, a special case of
the Dirac equation in dimension 1+3. The latter is a system of four
complex equations for four complex unknowns and if one looks for
solutions which do not depend on $x^3$ then this system splits into
a pair of systems (\ref{Dirac equation}).

Throughout this paper all fields are assumed to be infinitely smooth
with no assumptions on their behavior at infinity. We focus on
understanding the geometric meaning of equation (\ref{Dirac equation})
rather than on fitting it into the framework of operator theory.

We suggest a new geometric interpretation of equation
(\ref{Dirac equation}). The basic idea is to view our (1+2)-dimensional
spacetime as an elastic continuum whose material points can
experience no displacements, only rotations, with rotations of
different material points being totally independent. The idea of
rotating material points may seem exotic, however it has long been
accepted in continuum mechanics within the Cosserat theory of
elasticity \cite{Co}. This idea also lies at the heart of the theory
of \emph{teleparallelism} (=~absolute parallelism
=~fernparallelismus), a subject promoted by A.~Einstein and
\'E.~Cartan \cite{MR543192,MR2276051,unzicker-2005-}.
With regards to the latter it is interesting that
Cartan acknowledged \cite{cartan1} that he drew
inspiration from the `beautiful' work of the Cosserat brothers.

An elastic continuum with no displacements, only rotations, is, of
course, a limit case of Cosserat elasticity. The other limit case is
classical elasticity with displacements only and no
(micro)rotations.

Rotations of material points of the (1+2)-dimensional elastic
continuum are described mathematically by attaching to each
geometric point of Minkowski spacetime $\mathbb{M}^{1+2}$ an orthonormal
basis, which gives a field of orthonormal bases called the
\emph{frame} or \emph{coframe}, depending on whether one prefers
dealing with vectors or covectors. Our model will be built on the
basis of exterior calculus so for us it will be more natural to use
the coframe.

Our model is described in Section \ref{Our model}.
Subsequent sections contain mathematical analysis
culminating in Theorem \ref{main theorem}
(see Section \ref{Main result})
which establishes that our model is equivalent to the Dirac equation
(\ref{Dirac equation}).

The mathematical model presented in Section \ref{Our model} is quite
simple. However, seeing that this model generates the Dirac equation
(\ref{Dirac equation}) is not easy. The main difficulties are as
follows.
\begin{itemize}
\item
The dynamical variables in our model and the Dirac model are
different. We will overcome this difficulty by performing a
nonlinear change of dynamical variables given by the explicit
formulas (\ref{formula for density})--(\ref{formula for coframe elements 1 and 2}).
\item
We incorporate mass and electromagnetic field into our model by
means of a Kaluza--Klein extension, i.e. by adding an extra spatial
dimension and then separating out the extra coordinate $x^3$. Now,
our field equation (Euler--Lagrange equation) will turn out to be
nonlinear so the fact that it admits separation of variables is
nontrivial. We will establish separation of variables by performing
explicit calculations. We suspect that the underlying
group-theoretic reason for our nonlinear field equation admitting
separation of variables is the fact that our model is
$\mathrm{U}(1)$-invariant, i.e.~it is invariant under the
multiplication of the spinor field by a complex constant of modulus
1. Hence, it is feasible that one could perform the separation of
variables arguments without writing down the explicit form of the
field equation.
\item
Our field equation will be second order so it is unclear how it can
be reduced to a first order equation (\ref{Dirac equation}). This
issue will be addressed in Appendix~\ref{Nonlinear second order equations}.
Namely, in this appendix we
prove an abstract lemma showing that a certain class of nonlinear
second order partial differential equations reduces to pairs of
linear first order equations. To our knowledge, this abstract lemma is a new
result.
\end{itemize}

Our paper is a development of the publication \cite{1001.4726} where
a similar model was suggested for a massless fermion (neutrino).

\section{Our model}
\label{Our model}

The coframe $\vartheta$ is a triple of orthonormal covector fields
$\vartheta^j$, $j=0,1,2$, in $\mathbb{M}^{1+2}$. Each covector field
$\vartheta^j$ can be written more explicitly as
$\vartheta^j{}_\alpha$ where the tensor index $\alpha=0,1,2$
enumerates the components.
Of course, orthonormality is understood
in the Lorentzian sense: the inner product
$\vartheta^j\cdot\vartheta^k=g^{\alpha\beta}\vartheta^j{}_\alpha\vartheta^k{}_\beta$
is $\ -1\ $ if $j=k=0$, $\ +1\ $ if $j=k=1$ or $j=k=2$, and zero otherwise.

The orthonormality condition for the
coframe can be represented as a single tensor identity
\begin{equation}
\label{constraint for coframe}
g=o_{jk}\vartheta^j\otimes\vartheta^k
\end{equation}
where
\begin{equation}
\label{formula for os}
o_{jk}=o^{jk}:=\operatorname{diag}(-1,+1,+1).
\end{equation}
For the sake of clarity we repeat formula~(\ref{constraint for coframe})
giving tensor indices explicitly and performing summation over frame
indices explicitly:
$
g_{\alpha\beta}
=-\vartheta^0{}_\alpha\vartheta^0{}_\beta
+\vartheta^1{}_\alpha\vartheta^1{}_\beta
+\vartheta^2{}_\alpha\vartheta^2{}_\beta
$
where $\alpha$ and $\beta$ run through the values $0,1,2$.
We view the identity (\ref{constraint for coframe}) as a kinematic
constraint: the covector fields
$\vartheta^j$ are chosen so that they
satisfy~(\ref{constraint for coframe}), which leaves us with three real
degrees of freedom at every point of $\mathbb{M}^{1+2}$.
If one views $\vartheta^j{}_\alpha$ as a $3\times3$ real
matrix-function, then condition (\ref{constraint for coframe})
means that this matrix-function is
pseudo-orthogonal, i.e.~orthogonal with respect to the Lorentzian inner product.

We choose to work with coframes
satisfying conditions
\begin{equation}
\label{conditions for proper coframe}
\det\vartheta^j{}_\alpha=+1>0,
\qquad
\vartheta^0{}_0>0
\end{equation}
which single out coframes that can be obtained from the
trivial (aligned with coordinate lines) coframe
$\vartheta^j{}_\alpha=\delta^j{}_\alpha$
by proper Lorentz transformations.

As dynamical variables in our model we choose the coframe
$\vartheta$ and a positive density $\rho$. Our coframe and density
are functions of coordinates $x^\alpha$, $\alpha=0,1,2$,
in $\mathbb{M}^{1+2}$.
At a physical level, making the density $\rho$ a
dynamical variable means that we view our continuum more like a
fluid rather than a solid: we allow the material to
redistribute itself so that it finds its equilibrium distribution.
Note that the total number of real dynamical degrees of freedom
contained in the coframe $\vartheta$ and positive density $\rho$
is four, exactly as in a two-component complex-valued spinor field $\eta$.

In order to incorporate into our model mass and electromagnetic
field we perform a Kaluza--Klein extension: we extend our original
(1+2)-dimensional Minkowski spacetime $\mathbb{M}^{1+2}$
to (1+3)-dimensional Minkowski spacetime $\mathbb{M}^{1+3}$
by adding the extra spatial coordinate $x^3$.
The metric on $\mathbb{M}^{1+3}$ is
$\mathbf{g}_{{\bm\alpha}{\bm\beta}}=\operatorname{diag}(-1,+1,+1,+1)$.
Here and further on we use \textbf{bold} type for extended
quantities. Say, the use of bold type in the tensor indices of
$\mathbf{g}_{{\bm\alpha}{\bm\beta}}$
indicates that ${\bm{\alpha}}$ and ${\bm{\beta}}$
run through the values $0,1,2,3$.

We extend our coframe as
\begin{equation}
\label{extension of coframe}
{\bm\vartheta}^j{}_{\bm\alpha}=
\begin{pmatrix}
\vartheta^j{}_\alpha\\0
\end{pmatrix},
\qquad j=0,1,2,
\qquad
{\bm\vartheta}^3{}_{\bm\alpha}=
\begin{pmatrix}
0_\alpha\\1
\end{pmatrix}
\end{equation}
where the bold tensor index ${\bm\alpha}$ runs through the values
$0,1,2,3$, whereas its non-bold counterpart $\alpha$ runs through
the values $0,1,2$. In particular, the $0_\alpha$
in formula (\ref{extension of coframe}) stands for a column
of three zeros.

Our original (1+2)-dimensional coframe $\vartheta$, which was initially
a function of $(x^0,x^1,x^2)$ only, is now allowed to depend on
$x^3$ in an arbitrary way, as long as the kinematic constraint
(\ref{constraint for coframe}) is maintained. Our only
restriction on the choice of extended (1+3)-dimensional coframe
$\bm\vartheta$ is the condition that the last element of the coframe
is prescribed as the conormal to the original Minkowski spacetime $\mathbb{M}^{1+2}$,
see last formula (\ref{extension of coframe}).

We also extend our positive density $\rho$ allowing arbitrary
dependence on $x^3$.
We retain the non-bold type for the extended $\rho$.

The coframe elements ${\bm\vartheta}^{\mathbf{j}}$ are different at
different points $x\in\mathbb{M}^{1+3}$ and this causes
deformations. As a measure of these
``rotational deformations'' we choose \emph{axial torsion} which is
the 3-form defined by the formula
\begin{equation}
\label{definition of axial torsion extended}
\mathbf{T}^\mathrm{ax}:=\frac13\mathbf{o}_{\mathbf{j}\mathbf{k}}
{\bm\vartheta}^{\mathbf{j}}\wedge d{\bm\vartheta}^{\mathbf{k}}
\end{equation}
where $\mathbf{o}_{\mathbf{j}\mathbf{k}}
=\mathbf{o}^{\mathbf{j}\mathbf{k}}:=\mathrm{diag}(-1,+1,+1,+1)$
(compare with formula (\ref{formula for os}))
and $\,d\,$ denotes the exterior derivative on $\mathbb{M}^{1+3}$.
Here ``torsion'' stands for ``torsion of the teleparallel connection''
with ``teleparallel connection'' defined by the condition that
the covariant derivative of each coframe element ${\bm\vartheta}^{\mathbf{j}}$ is
zero; see Appendix A of \cite{MR2573111} for a concise exposition.
``Axial torsion'' is the totally antisymmetric part of the
torsion tensor (\ref{definition of torsion extended}).

We choose the basic Lagrangian density of our mathematical model as
\begin{equation}
\label{basic Lagrangian density}
L({\bm\vartheta},\rho):=\|\mathbf{T}^\mathrm{ax}\|^2\rho
\end{equation}
where
$\|\mathbf{T}^\mathrm{ax}\|^2=\frac1{3!}
\mathbf{T}^\mathrm{ax}_{{\bm\alpha}{\bm\beta}{\bm\gamma}}
\mathbf{T}^\mathrm{ax}_{{\bm\kappa}{\bm\lambda}{\bm\mu}}
\mathbf{g}^{{\bm\alpha}{\bm\kappa}}
\mathbf{g}^{{\bm\beta}{\bm\lambda}}
\mathbf{g}^{{\bm\gamma}{\bm\mu}}$.
The main motivation behind the choice of
Lagrangian density (\ref{basic Lagrangian density})
is the fact that it is conformally invariant:
it does not change if we rescale the coframe as
${\bm\vartheta}^{\mathbf{j}}\mapsto e^{\mathbf{h}}{\bm\vartheta}^{\mathbf{j}}$,
metric as
$\mathbf{g}_{{\bm\alpha}{\bm\beta}}\mapsto
e^{2\mathbf{h}}\mathbf{g}_{{\bm\alpha}{\bm\beta}}$
and density as
$\rho\mapsto e^{2\mathbf{h}}\rho$
where
$\mathbf{h}:\mathbb{M}^{1+3}\to\mathbb{R}$ is an arbitrary scalar function.
At this point it is important to note that out that our Kaluza--Klein
extension procedure does not actually allow for conformal rescalings
because the last formula (\ref{extension of coframe}) is very specific.
Thus, our logic is that we choose a Lagrangian density
(\ref{basic Lagrangian density}) which \emph{would be} conformally invariant
if not for the prescriptive nature of the Kaluza--Klein construction.
This is in line with the view that mass breaks conformal invariance.
The electron mass $m$ will appear below in formulas
(\ref{periodicity condition coframe}) and
(\ref{periodicity condition density}).

Substituting (\ref{extension of coframe}) into
(\ref{definition of axial torsion extended}) we get
\begin{equation}
\label{definition of axial torsion extended simplified}
\mathbf{T}^\mathrm{ax}=T^\mathrm{ax}-
{\bm\vartheta}^3\wedge D_3\vartheta
\end{equation}
where
\begin{equation}
\label{definition of axial torsion}
T^\mathrm{ax}:=\frac13o_{jk}
\vartheta^j\wedge d\vartheta^k
\end{equation}
is the axial torsion in original (1+2)-dimensional spacetime
(with $\,d\,$ now denoting the exterior derivative on $\mathbb{M}^{1+2}$)
and $D_3\vartheta$ is the 2-form
\begin{equation}
\label{definition of dot vartheta}
D_3\vartheta:=\frac13o_{jk}
\vartheta^j\wedge\partial_3\vartheta^k.
\end{equation}
The 2-form $D_3\vartheta$ characterizes the rotation of
the coframe $\vartheta$ as we move along the coordinate $x^3$ and is,
in effect, an analogue of angular velocity.

Substituting (\ref{definition of axial torsion extended simplified}) into
(\ref{basic Lagrangian density}) we rewrite our basic Lagrangian density as
\begin{equation}
\label{basic Lagrangian density 1}
L(\vartheta,\rho):=
(\|T^\mathrm{ax}\|^2+\|D_3\vartheta\|^2)\rho.
\end{equation}

We now incorporate the electron mass $m$ into our model by imposing
the periodicity conditions
\begin{equation}
\label{periodicity condition coframe}
\vartheta(x^0,x^1,x^2,x^3+\pi/m)=\vartheta(x^0,x^1,x^2,x^3),
\end{equation}
\begin{equation}
\label{periodicity condition density}
\rho(x^0,x^1,x^2,x^3+\pi/m)=\rho(x^0,x^1,x^2,x^3).
\end{equation}
Conditions (\ref{periodicity condition coframe}) and
(\ref{periodicity condition density}) mean that we
make the coordinate $x^3$ cyclic with period~$\,\frac\pi m\,$.
In other words, we effectively roll up
our third spatial dimension into a circle of radius $\frac1{2m}\,$.

Finally, we incorporate the prescribed electromagnetic (co)vector potential
$A$ into our model by formally mixing up the partial derivatives appearing
in the definition of axial torsion (\ref{definition of axial torsion})
as
\begin{equation}
\label{mixing up the derivatives}
\partial_\alpha\mapsto\partial_\alpha+m^{-1}\!A_\alpha\partial_3\,,
\qquad\alpha=0,1,2\,.
\end{equation}
As a result, our Lagrangian density
(\ref{basic Lagrangian density 1}) turns into
\begin{equation}
\label{our Lagrangian density with A}
L(\vartheta,\rho):=
(\|T_A^\mathrm{ax}\|^2+\|D_3\vartheta\|^2)\rho,
\end{equation}
where
\begin{equation}
\label{definition of axial torsion with A}
T_A^\mathrm{ax}:=
T^\mathrm{ax}-m^{-1}A\wedge D_3\vartheta.
\end{equation}

Let us summarize the above construction. The Lagrangian density that
we shall be studying is given by formula
(\ref{our Lagrangian density with A})
where the 3-form $T_A^\mathrm{ax}$ and 2-form $D_3\vartheta$ are
defined by formulas
(\ref{definition of axial torsion}),
(\ref{definition of dot vartheta})
and (\ref{definition of axial torsion with A}).
The corresponding action (variational functional) is
\begin{equation}
\label{our action}
S(\vartheta,\rho):=
\int_{\mathbb{M}^{1+3}}L(\vartheta,\rho)\,dx^0dx^1dx^2dx^3\,;
\end{equation}
of course, the integral in (\ref{our action}) need
not converge as we will be using it only for the purpose of deriving
field equations (Euler--Lagrange equations).
Our dynamical variables are the coframe $\vartheta$ and density
$\rho$ which live in the original (1+2)-dimensional spacetime
but depend on the extra spatial coordinate $x^3$.
We seek solutions which are periodic in $x^3$,
see formulas
(\ref{periodicity condition coframe}) and
(\ref{periodicity condition density}).

Our field equations are obtained by
varying the action (\ref{our action}) with respect to the
coframe $\vartheta$ and density $\rho$. Varying with respect to the
density $\rho$ is easy: this gives the field equation
$\|T_A^\mathrm{ax}\|^2+\|D_3\vartheta\|^2=0$
which is equivalent to
$L(\vartheta,\rho)=0$. Varying with respect to the coframe
$\vartheta$ is more difficult because we have to maintain the kinematic
constraint (\ref{constraint for coframe}).
A technique for varying the coframe with kinematic constraint
(\ref{constraint for coframe}) was described in Appendix~B of
\cite{MR2573111} but we do not use it in the current paper.

\section{Switching to the language of spinors}
\label{Switching to the language of spinors}

As pointed out in the previous section, varying the
coframe subject to the kinematic constraint
(\ref{constraint for coframe}) is not an easy task.
This technical difficulty can be overcome by switching
to a different dynamical variable. Namely, it is known
that in dimension 1+2 a coframe $\vartheta$ and a
positive density $\rho$ are equivalent to a 2-component
complex-valued spinor field
$
\xi=\xi^a=
\begin{pmatrix}
\xi^1\\
\xi^2
\end{pmatrix}
$
satisfying the inequality
\begin{equation}
\label{density is positive}
\bar\xi^{\dot a}\sigma_{3\dot ab}\xi^b>0.
\end{equation}
The explicit formulas establishing this equivalence are
\begin{equation}
\label{formula for density}
\rho=\bar\xi^{\dot a}\sigma_{3\dot ab}\xi^b,
\end{equation}
\begin{equation}
\label{formula for coframe element 0}
\vartheta^0{}_\alpha=\rho^{-1}
\bar\xi^{\dot a}\sigma_{\alpha\dot ab}\xi^b,
\end{equation}
\begin{equation}
\label{formula for coframe elements 1 and 2}
(\vartheta^1+i\vartheta^2)_\alpha=\rho^{-1}
\epsilon^{\dot c\dot b}\sigma_{3\dot ba}\xi^a\sigma_{\alpha\dot cd}\xi^d.
\end{equation}
Here $\sigma$ are Pauli matrices and $\epsilon$ is ``metric spinor''
(see (\ref{metric spinor})--(\ref{Pauli matrices spatial})),
the free tensor index $\alpha$ runs through the values
$0,1,2$, and the spinor summation indices run through the values
$1,2$ or $\dot1,\dot2$. The advantage of switching to a spinor
field $\xi$ is that there are no kinematic constraints on its
components, so the derivation of field equations becomes
straightforward.

Formulas
(\ref{formula for density})--(\ref{formula for coframe elements 1 and 2})
are a variant of those from \cite{1001.4726}: in \cite{1001.4726}
these formulas were written for dimension 3,
i.e.~for 3-dimensional Euclidean space,
whereas in the current
paper we write them for dimension 1+2,
i.e.~for (1+2)-dimensional Minkowski spacetime.
Both the formulas from \cite{1001.4726} and formulas
(\ref{formula for density})--(\ref{formula for coframe elements 1 and 2})
are a special case of those from \cite{MR0332092}.

\begin{remark}
\label{remark about indeterminate sign}
The right-hand sides of formulas
(\ref{formula for density})--(\ref{formula for coframe elements 1 and 2})
are invariant under the change of sign of $\xi$. Hence, the
correspondence between coframe and positive density on the one hand
and spinor field satisfying condition (\ref{density is positive}) on
the other is one to two. A spinor field is, effectively, a square
root of a coframe and a density. The fact that the spinor field has
indeterminate sign does not cause problems as long as we work on a
simply connected open set\footnote{Here and further on the notions of openness and
connectedness of subsets of $\mathbb{M}^{1+2}$
are understood in the Euclidean sense, i.e.~in terms of a positive
3-dimensional metric.},
such as the whole Minkowski space
$\mathbb{M}^{1+2}$.
Note that a similar issue (extraction of
a single-valued ``square root'' of a tensor)
arises in the mathematical theory of liquid crystals~\cite{ball_and_zarnescu}.
\end{remark}

We now need to express the differential forms
(\ref{definition of axial torsion}),
(\ref{definition of dot vartheta})
and
(\ref{definition of axial torsion with A})
via the spinor field~$\xi$.
This is done by direct substitution
of formulas
(\ref{formula for density})--(\ref{formula for coframe elements 1 and 2})
giving
\begin{equation}
\label{axial torsion via spinor}
*T^\mathrm{ax}=-
\frac{
2i(
\bar\xi^{\dot a}\sigma^\alpha{}_{\dot ab}\partial_\alpha\xi^b
-
\xi^b\sigma^\alpha{}_{\dot ab}\partial_\alpha\bar\xi^{\dot a}
)
}
{
3\bar\xi^{\dot c}\sigma_{3\dot cd}\xi^d
}\,,
\end{equation}
\begin{equation}
\label{dot vartheta via spinor}
(*D_3\vartheta)_\alpha=
\frac{
2i(
\bar\xi^{\dot a}\sigma_{\alpha\dot ab}\partial_3\xi^b
-
\xi^b\sigma_{\alpha\dot ab}\partial_3\bar\xi^{\dot a}
)
}
{
3\bar\xi^{\dot c}\sigma_{3\dot cd}\xi^d
}\,,
\end{equation}
\begin{equation}
\label{axial torsion via spinor with A}
*T_A^\mathrm{ax}=-
\frac{
2i(
\bar\xi^{\dot a}\sigma^\alpha{}_{\dot ab}
(\partial_\alpha+m^{-1}\!A_\alpha\partial_3)\xi^b
-
\xi^b\sigma^\alpha{}_{\dot ab}
(\partial_\alpha+m^{-1}\!A_\alpha\partial_3)\bar\xi^{\dot a}
)
}
{
3\bar\xi^{\dot c}\sigma_{3\dot cd}\xi^d
}\,.
\end{equation}
The tensor summation index $\alpha$ in formulas
(\ref{axial torsion via spinor}) and
(\ref{axial torsion via spinor with A})
and the free tensor index
$\alpha$ in formula (\ref{dot vartheta via spinor}) run through
the values $0,1,2$.
Formulas
(\ref{axial torsion via spinor}) and (\ref{dot vartheta via spinor})
are, of course, a variant of those from \cite{1001.4726}:
we simply turned 3-dimensional Euclidean space into (1+2)-dimensional Minkowski
space and replaced the extra coordinate $x^0$ with the extra coordinate $x^3$.

Substituting formulas
(\ref{axial torsion via spinor with A}) and (\ref{dot vartheta via spinor})
into
(\ref{our Lagrangian density with A})
we arrive at the following self-contained explicit
spinor representation of our Lagrangian density
\begin{multline}
\label{our Lagrangian density with A via spinor}
L(\xi)=
-\frac{4}{9\bar\xi^{\dot c}\sigma_{3\dot cd}\xi^d}
\\
\Bigl(
\bigl[
i(
\bar\xi^{\dot a}\sigma^\alpha{}_{\dot ab}
(\partial_\alpha+m^{-1}\!A_\alpha\partial_3)\xi^b
-
\xi^b\sigma^\alpha{}_{\dot ab}
(\partial_\alpha+m^{-1}\!A_\alpha\partial_3)\bar\xi^{\dot a}
)
\bigr]^2
\\
+\bigl\|
i(
\bar\xi^{\dot a}\sigma_{\alpha\dot ab}
\partial_3\xi^b
-
\xi^b\sigma_{\alpha\dot ab}
\partial_3\bar\xi^{\dot a}
)
\bigr\|^2
\Bigr).
\end{multline}
Here and
further on we write our Lagrangian density and our action as
$L(\xi)$ and $S(\xi)$ rather than $L(\vartheta,\rho)$ and
$S(\vartheta,\rho)$, thus indicating that we have switched to
spinors. The spinor field $\xi$ satisfying condition
(\ref{density is positive})
is the new dynamical variable.

The field equation for our Lagrangian
density~(\ref{our Lagrangian density with A via spinor}) is
\begin{multline}
\label{field equation}
\!\!\!\!\!
\frac{4i}3\Bigl(
(*T_A^\mathrm{ax})\sigma^\alpha{}_{\dot ab}
(\partial_\alpha+m^{-1}\!A_\alpha\partial_3)\xi^b
+\sigma^\alpha{}_{\dot ab}
(\partial_\alpha+m^{-1}\!A_\alpha\partial_3)
((*T_A^\mathrm{ax})\xi^b)
\\
-(*D_3\vartheta)_\alpha\sigma^\alpha{}_{\dot ab}
\partial_3\xi^b
-\sigma^\alpha{}_{\dot ab}
\partial_3((*D_3\vartheta)_\alpha\xi^b)
\Bigr)
-\rho^{-1}L\sigma_{3\dot ab}\xi^b
=0
\end{multline}
where the quantities
$*T_A^\mathrm{ax}$,
$*D_3\vartheta$,
$\rho$ and
$L$
are expressed via the spinor field $\xi$ in accordance with formulas
(\ref{axial torsion via spinor with A}),
(\ref{dot vartheta via spinor}),
(\ref{formula for density}) and
(\ref{our Lagrangian density with A via spinor}).

We seek solutions
of the field equation (\ref{field equation})
which satisfy the periodicity condition
\begin{equation}
\label{periodicity condition spinor}
\xi(x^0,x^1,x^2,x^3+\pi/m)=\xi(x^0,x^1,x^2,x^3),
\end{equation}
or the antiperiodicity condition
\begin{equation}
\label{antiperiodicity condition spinor}
\xi(x^0,x^1,x^2,x^3+\pi/m)=-\xi(x^0,x^1,x^2,x^3).
\end{equation}
The above periodicity/antiperiodicity conditions are our original
periodicity conditions
(\ref{periodicity condition coframe}) and (\ref{periodicity condition density})
rewritten in terms of the spinor field.
The splitting into periodicity/antiperiodicity occurs because the
spinor field corresponding to a coframe and a density is determined
uniquely modulo sign, see Remark \ref{remark about indeterminate sign}.

\section{Separating out the coordinate $x^3$}
\label{Separating out the coordinate x3}

Our field equation (\ref{field equation}) is highly nonlinear
and one does not expect it to admit separation of variables.
Nevertheless,
we seek solutions of the form
\begin{equation}
\label{xi in terms of eta}
\xi(x^0,x^1,x^2,x^3)=
\eta(x^0,x^1,x^2)\,e^{\mp imx^3}.
\end{equation}
Note that such solutions automatically satisfy the antiperiodicity
condition (\ref{antiperiodicity condition spinor}):
the coframe corresponding to a spinor field of the
form (\ref{xi in terms of eta}) experiences one full turn
(clockwise or anticklockwise) in
the $(\vartheta^1,\vartheta^2)$-plane as $x^3$ runs from 0 to~$\frac\pi m$.

Substituting formula
(\ref{xi in terms of eta})
into
(\ref{axial torsion via spinor with A}),
(\ref{dot vartheta via spinor}),
(\ref{formula for density}) and
(\ref{our Lagrangian density with A via spinor})
we get
\begin{equation}
\label{axial torsion via spinor with A for eta}
*T_{A\pm}^\mathrm{ax}=-
\frac{
2(
\bar\eta^{\dot a}\sigma^\alpha{}_{\dot ab}
(i\partial\pm A)_\alpha\eta^b
-
\eta^b\sigma^\alpha{}_{\dot ab}
(i\partial\mp A)_\alpha\bar\eta^{\dot a}
)
}
{
3\bar\eta^{\dot c}\sigma_{3\dot cd}\eta^d
}\,,
\end{equation}
\begin{equation}
\label{dot vartheta via spinor with A for eta}
(*D_3\vartheta)_\alpha=
\pm\frac{
4m\bar\eta^{\dot a}\sigma_{\alpha\dot ab}\eta^b
}
{
3\bar\eta^{\dot c}\sigma_{3\dot cd}\eta^d
}\,,
\end{equation}
\begin{equation}
\label{formula for density for eta}
\rho=\bar\eta^{\dot a}\sigma_{3\dot ab}\eta^b,
\end{equation}
\begin{multline}
\label{our Lagrangian density with A via spinor for eta}
L_\pm(\eta)=
-\frac{16}{9\bar\eta^{\dot c}\sigma_{3\dot cd}\eta^d}
\\
\Bigl(
\bigl[
\tfrac12(
\bar\eta^{\dot a}\sigma^\alpha{}_{\dot ab}
(i\partial\pm A)_\alpha\eta^b
-
\eta^b\sigma^\alpha{}_{\dot ab}
(i\partial\mp A)_\alpha\bar\eta^{\dot a}
)
\bigr]^2
-
(m\bar\eta^{\dot a}\sigma_{3\dot ab}\eta^b)^2
\Bigr)
\end{multline}
where the signs agree with those in (\ref{xi in terms of eta})
(upper sign corresponds to upper sign and lower sign corresponds to
lower sign).

Note that the quantities
(\ref{axial torsion via spinor with A for eta})--(\ref{our Lagrangian density with A via spinor for eta})
do not depend on $x^3$, which simplifies the next step:
substituting
(\ref{xi in terms of eta})
into our field equation
(\ref{field equation})
and dividing through by the common factor
$e^{\mp imx^3}$
we get
\begin{multline}
\label{field equation for eta}
\!\!\!\!\!
\frac43\Bigl(
(*T_{A\pm}^\mathrm{ax})\sigma^\alpha{}_{\dot ab}
(i\partial\pm A)_\alpha\eta^b
+\sigma^\alpha{}_{\dot ab}
(i\partial\pm A)_\alpha
((*T_{A\pm}^\mathrm{ax})\eta^b)
\Bigr)
\\
+\frac{32m^2}9\sigma^3{}_{\dot ab}\eta^b
-\rho^{-1}L_\pm\sigma_{3\dot ab}\eta^b
=0.
\end{multline}

Observe that formulas
(\ref{axial torsion via spinor with A for eta})--(\ref{field equation for eta})
do not contain $x^3$.
Thus, we have shown that our
field equation
(\ref{field equation})
admits separation of
variables, i.e.~one can seek solutions of the
form~(\ref{xi in terms of eta}).

Consider now the action
\begin{equation}
\label{our action for eta}
S_\pm(\eta):=
\int_{\mathbb{M}^{1+2}}L_\pm(\eta)\,dx^0dx^1dx^2
\end{equation}
where $L_\pm(\eta)$ is the Lagrangian density
(\ref{our Lagrangian density with A via spinor for eta}).
It is easy to see that equation
(\ref{field equation for eta}) is the field equation (Euler--Lagrange equation)
for the action~(\ref{our action for eta}).

In the remainder of the paper we do not use the explicit form of the
field equation (\ref{field equation for eta}), dealing only
with the Lagrangian density
(\ref{our Lagrangian density with A via spinor for eta})
and action (\ref{our action for eta}).
We needed the explicit form of field equations,
(\ref{field equation}) and
(\ref{field equation for eta}),
only to justify separation of variables.

We give for reference a more compact representation of our
Lagrangian density
(\ref{our Lagrangian density with A via spinor for eta})
in terms of axial torsion $T_{A\pm}^\mathrm{ax}$
(see formula (\ref{axial torsion via spinor with A for eta}))
and density~$\rho$
(see formula~(\ref{formula for density for eta})):
\begin{equation}
\label{our Lagrangian density plus minus compact}
L_\pm(\eta)=-\Bigl(
\bigl(*T_{A\pm}^\mathrm{ax}\bigr)^2-\frac{16}9m^2
\Bigr)\rho\,.
\end{equation}
Of course, formula (\ref{our Lagrangian density plus minus compact})
is our original formula (\ref{our Lagrangian density with A})
with $x^3$ separated out.
The choice of dynamical variables in the Lagrangian density
(\ref{our Lagrangian density plus minus compact}) is up to the user: one can either
use the $x^3$-independent spinor field $\eta$ or, equivalently,
the corresponding $x^3$-independent coframe and $x^3$-independent density
(the latter are related to $\eta$ by formulas
(\ref{formula for density})--(\ref{formula for coframe elements 1 and 2})
with $\xi$ replaced by $\eta$).

\section{Main result}
\label{Main result}

Let $D_{rs}$ be the linear differential operator mapping
undotted spinor fields into dotted spinor fields in accordance with formula
\begin{equation}
\label{Dirac operator}
\eta\ \mapsto
\
D_{rs}\eta=
\sigma^\alpha{}_{\dot ab}(i\partial_\alpha+rA_\alpha)\eta^b
+sm\sigma^3{}_{\dot ab}\eta^b
\end{equation}
where the tensor summation index $\alpha$ runs through the values
$0,1,2$ and the letters $r$ and $s$ take, independently,
symbolic values $\pm$ (as in $D_{rs}$)
or numerical values $\pm1$ (as in the RHS of formula
(\ref{Dirac operator})),
depending on the context.

The main result of our paper is

\begin{theorem}
\label{main theorem}
Let $\Omega$ be an open
subset of $\mathbb{M}^{1+2}$
and let $\eta:\Omega\to\mathbb{C}^2$ be a spinor field satisfying
the condition
\begin{equation}
\label{density is positive for eta}
\bar\eta^{\dot a}\sigma_{3\dot ab}\eta^b>0
\end{equation}
(compare with (\ref{density is positive})).
Then $\eta$
is a solution of the field equation
for the Lagrangian density $L_+$
if and only if it is a solution of
the Dirac equation $D_{++}\eta=0$
or the Dirac equation $D_{+-}\eta=0$,
and a solution of the field equation
for the Lagrangian density $L_-$
if and only if it is a solution of
the Dirac equation $D_{-+}\eta=0$
or the Dirac equation $D_{--}\eta=0$.
\end{theorem}

\begin{proof}
Put
\begin{equation}
\label{Dirac Lagrangian density}
L_{rs}(\eta):=\frac12
\bigl[
\bar\eta^{\dot a}\sigma^\alpha{}_{\dot ab}(i\partial_\alpha+rA_\alpha)\eta^b
-\eta^b\sigma^\alpha{}_{\dot ab}(i\partial_\alpha-rA_\alpha)\bar\eta^{\dot a}
\bigr]
+sm\bar\eta^{\dot a}\sigma^3{}_{\dot ab}\eta^b.
\end{equation}
This is the Lagrangian density for the Dirac equation
$D_{rs}\eta=0$.
Formula (\ref{Dirac Lagrangian density})
can be rewritten in more compact form as
\begin{equation}
\label{Dirac Lagrangian density compact}
L_{rs}(\eta)
=\Bigl(
-\frac34*T_{Ar}^\mathrm{ax}+sm
\Bigr)\rho
\end{equation}
where
$*T_{Ar}^\mathrm{ax}$, $r=\pm$, is the Hodge dual of axial torsion
defined by formula (\ref{axial torsion via spinor with A for eta})
and $\rho$ is the density
defined by formula (\ref{formula for density for eta}).
Comparing formulas
(\ref{our Lagrangian density plus minus compact})
and
(\ref{Dirac Lagrangian density compact})
we get
\begin{equation}
\label{factorization formula}
L_r(\eta)=-\frac{32m}9\,
\frac{L_{r+}(\eta)\,L_{r-}(\eta)}{L_{r+}(\eta)-L_{r-}(\eta)}\,.
\end{equation}
Note that the denominator in the above formula is nonzero because
condition (\ref{density is positive for eta}) can be equivalently
rewritten as $L_{r+}(\eta)>L_{r-}(\eta)$.

The result now follows from formula
(\ref{factorization formula}) and Lemma~\ref{Nonlinear second order equations lemma}
(see Appendix~\ref{Nonlinear second order equations}).~$\square$
\end{proof}

\section{The sign in the inequality (\ref{density is positive})}
\label{The sign in the inequality}

In Section \ref{Switching to the language of spinors}, when
switching to the language of spinors, we chose to work with spinor
fields $\xi$ satisfying the inequality (\ref{density is positive}).
It is natural to ask the question what happens if we choose to work
with spinor fields $\tilde\xi$ satisfying the inequality
\begin{equation}
\label{density is negative}
\bar{\tilde\xi}^{\dot a}\sigma_{3\dot ab}{\tilde\xi}^b<0.
\end{equation}
One can check that in this case all our arguments can be repeated
with minor changes. Namely, in dimension 1+2 a coframe $\vartheta$
and a positive density $\rho$ are equivalent to a 2-component
complex-valued spinor field $\,\tilde\xi$ satisfying
the inequality~(\ref{density is negative}),
with this equivalence described
by a slightly modified version of formulas
(\ref{formula for density})--(\ref{formula for coframe elements 1 and 2}).
In the end we get an analogue of Theorem \ref{main theorem} for such spinors.

In fact, there is no need to repeat our arguments because there is a
bijection between
spinor fields $\xi$ satisfying the inequality (\ref{density is positive})
and
spinor fields $\tilde\xi$ satisfying the inequality (\ref{density is negative}):
\begin{equation}
\label{bijection}
\xi\mapsto{\tilde\xi}^c=\epsilon^{cb}\sigma_{3\dot ab}\bar\xi^{\dot a},
\qquad
\tilde\xi\mapsto\xi^c=\epsilon^{cb}\sigma_{3\dot ab}\bar{\tilde\xi}^{\dot a}.
\end{equation}

We do not view the transformation (\ref{bijection}) as physically
significant because the primary dynamical variables in our model are
coframe and positive density, not the spinor field. We view the
spinor field merely as a convenient change of dynamical variables.
If two different spinor fields correspond to the same coframe and
positive density we interpret them as the same particle. In
group-theoretical language this means that our model is built on the
basis of the pseudo-orthogonal group $\mathrm{SO}(1,2)$ rather than
the spin group $\mathrm{Spin}(1,2)$.

\section{Plane wave solutions}
\label{Plane wave solutions}

In this section we construct a special class of explicit solutions
of the field equations for our Lagrangian density
(\ref{our Lagrangian density with A}).
This construction is presented, initially, in the language of spinors
and under the additional assumption that
the electromagnetic covector potential $A$ is zero.

We seek solutions of the form
\begin{equation}
\label{Plane wave solutions equation 1}
\xi(x^0,x^1,x^2,x^3)=
e^{-i(p\cdot x+rmx^3)}\zeta
\end{equation}
where $p=(p_0,p_2,p_3)$ is a real constant covector,
$r$ takes the values $\pm1$
and $\zeta\ne0$ is a constant spinor.
We shall call solutions of the type
(\ref{Plane wave solutions equation 1}) \emph{plane wave\/}.
In seeking plane wave solutions
what we are doing is separating out all the variables,
namely, the original variables $x=(x^0,x^1,x^2)$
(coordinates on $\mathbb{M}^{1+2}$) and the extra variable $x^3$
(Kaluza--Klein coordinate).

As usual, our spinor field $\xi$ is assumed to
satisfy the inequality (\ref{density is positive}).
As explained in Section \ref{The sign in the inequality},
this assumption does not lead to the loss of solutions.

Our field equation (\ref{field equation})
is highly nonlinear so it is
not \emph{a priori} clear that one can seek solutions in the form of plane
waves. However, plane wave solutions are a special case of
solutions of the type
(\ref{xi in terms of eta})
and these have already been analyzed in preceding sections.
Namely, Theorem~\ref{main theorem} gives us an algorithm for the calculation
of all plane wave solutions (\ref{Plane wave solutions equation 1})
by reducing the problem to Dirac equations
\begin{equation}
\label{Plane wave solutions equation 2}
D_{rs}\eta=0
\end{equation}
for the $x^3$-independent spinor field
\begin{equation}
\label{Plane wave solutions equation 3}
\eta(x^0,x^1,x^2)=e^{-ip\cdot x}\zeta.
\end{equation}
Here $r$ is the same as in formula
(\ref{Plane wave solutions equation 1}),
i.e.~a number taking the values $\pm1$,
and $s$ is another number, also taking, independently, the values $\pm1$.
By $D_{rs}$ we denote the differential operators
(\ref{Dirac operator}).

Clearly, a Dirac equation (\ref{Plane wave solutions equation 2})
has a nontrivial plane
wave solution $\eta$ if and only if the momentum $p$ satisfies the condition
$\|p\|^2+m^2=0$, so $p$ is timelike. Our model
is invariant under proper Lorentz transformations of coordinates
$(x^0,x^1,x^2)$ so without loss of generality
we can assume that
\begin{equation}
\label{Plane wave solutions equation 4}
p_1=p_2=0.
\end{equation}
Combining formulas
(\ref{Dirac operator}),
(\ref{Pauli matrix temporal}),
(\ref{Pauli matrices spatial}),
(\ref{Plane wave solutions equation 3})
and
(\ref{Plane wave solutions equation 4})
we see that
the Dirac equation (\ref{Plane wave solutions equation 2}) takes the form
\begin{equation}
\label{Plane wave solutions equation 5}
\begin{pmatrix}
-p_0+sm&0\\
0&-p_0-sm
\end{pmatrix}
\begin{pmatrix}
\zeta^1\\
\zeta^2
\end{pmatrix}=0.
\end{equation}
Equation (\ref{Plane wave solutions equation 3}) has a nontrivial
solution satisfying the inequality (\ref{density is positive})
only if
\begin{equation}
\label{Plane wave solutions equation 6}
p_0=sm
\end{equation}
with the corresponding $\zeta$ given, up to scaling by a nonzero
complex factor, by the formula
\begin{equation}
\label{Plane wave solutions equation 7}
\zeta^d=
\begin{pmatrix}
1\\
0
\end{pmatrix}.
\end{equation}

Combining formulas (\ref{Plane wave solutions equation 1}),
(\ref{Plane wave solutions equation 4}),
(\ref{Plane wave solutions equation 6}) and
(\ref{Plane wave solutions equation 7})
we conclude that our model admits, up to a proper Lorentz
transformation of the coordinate system in $\mathbb{M}^{1+2}$
and complex scaling, four plane
wave solutions and that these plane wave solutions are given by the
explicit formula
\begin{equation}
\label{Plane wave solutions equation 8}
\xi^d=
\begin{pmatrix}
1\\
0
\end{pmatrix}
e^{-im(sx^0+rx^3)}\,.
\end{equation}
Here the numbers $r$ and $s$ can, independently, take values $\pm1$.

Let us now rewrite the plane wave solutions
(\ref{Plane wave solutions equation 8})
in terms of our original dynamical variables, coframe $\vartheta$
and density $\rho$. Substituting formulae
(\ref{Pauli matrix temporal}),
(\ref{Pauli matrices spatial})
and
(\ref{Plane wave solutions equation 8})
into formulae
(\ref{formula for density})--(\ref{formula for coframe elements 1 and 2})
we get $\rho=1$, $\vartheta^0{}_\alpha=\delta^0{}_\alpha$ and
\begin{equation}
\label{Plane wave solutions equation 9}
\vartheta^1{}_\alpha
=\begin{pmatrix}
0\\
\cos2m(sx^0+rx^3)\\
\sin2m(sx^0+rx^3)
\end{pmatrix},
\qquad
\vartheta^2{}_\alpha
=\begin{pmatrix}
0\\
-\sin2m(sx^0+rx^3)\\
\cos2m(sx^0+rx^3)
\end{pmatrix}.
\end{equation}

In order to distinguish the two spins we fix $x^3$ and examine how
the covectors $\vartheta^1$ and $\vartheta^2$ evolve as a function
of time $x^0$. We say that spin is up if the rotation is
counterclockwise and spin is down if the rotation is clockwise.
Examination of formula~(\ref{Plane wave solutions equation 9})
shows that we have spin up if $s=+1$
and spin down if $s=-1$.

We will now establish which of the solutions
(\ref{Plane wave solutions equation 9})
describe the electron and which describe the positron.
Let us introduce a weak constant positive electric
field, $0<A_0<m$ and $A_1=A_2=0$.
Then we can repeat the calculation leading up to
formula (\ref{Plane wave solutions equation 9}), only now
we get
\begin{multline}
\label{Plane wave solutions equation 10}
\vartheta^1{}_\alpha
=\begin{pmatrix}
0\\
\cos2[(sm-rA_0)x^0+rmx^3]\\
\sin2[(sm-rA_0)x^0+rmx^3]
\end{pmatrix},
\\
\vartheta^2{}_\alpha
=\begin{pmatrix}
0\\
-\sin2[(sm-rA_0)x^0+rmx^3]\\
\cos2[(sm-rA_0)x^0+rmx^3]
\end{pmatrix}.
\end{multline}
We define quantum mechanical energy as
\begin{equation}
\label{Plane wave solutions equation 11}
\varepsilon:=|sm-rA_0|
\end{equation}
which is half the angular frequency (as a function of time $x^0$)
of the solution~(\ref{Plane wave solutions equation 10}).
We say that we are dealing with an electron if $\varepsilon<m$ and
with a positron if $\varepsilon>m$. Examination of formula
(\ref{Plane wave solutions equation 11}) shows that we are looking
at an electron if the signs of $r$ and $s$ are the same and
at a positron if the signs of $r$ and $s$ are opposite.
This means that
the electron is described by a wave traveling in the negative $x^3$-direction
whereas
the positron is described by a wave traveling in the positive $x^3$-direction.

Our classification of plane wave solutions is summarized in Table
\ref{table 1}.

\begin{table}[htb]
\caption{Classification of solutions (\ref{Plane wave solutions equation 9})}
\label{table 1}
\begin{tabular}{c|c|c}
{}&$s=+1$&$s=-1$\\
\hline
$r=+1$&Electron with spin up&Positron with spin down\\
\hline
$r=-1$&Positron with spin up&Electron with spin down
\end{tabular}
\end{table}

\section{Discussion}
\label{Discussion}

\subsection{Distinguishing the electron from the positron}

The mathematical model presented in this paper allows us to clearly
distinguish the electron from the positron.
This is
achieved by using the coframe and positive density as our primary
dynamical variables rather than the more traditional spinor field.
In other words,
as explained in the end of Section~\ref{The sign in the inequality},
our model is built on the basis of the pseudo-orthogonal
group $\mathrm{SO}(1,2)$ rather than the spin group
$\mathrm{Spin}(1,2)$.

\subsection{Problem of vanishing density}

The only technical assumption in our analysis is that the density
$\rho$ does not vanish. Rephrased in terms of the spinor field, this
assumption reads as
\begin{equation}
\label{density is nonzero}
\bar\xi^{\dot a}\sigma_{3\dot ab}\xi^b\ne0,
\end{equation}
compare with (\ref{density is positive}) and (\ref{density is negative}).
We do not know how to drop the
assumption (\ref{density is nonzero}).

\subsection{Curved spacetime}
\label{Curved spacetime}

One of the advantages of our mathematical model is that it does not
use covariant differentiation (only exterior differentiation) so the
generalization to the case of a
curved\footnote{Here ``curved'' refers to the curvature
of the Levi-Civita connection generated by the metric~$g$,
as is customary in General Relativity.}
(1+2)-dimensional spacetime
is absolutely straightforward. Covariant derivatives appear only
when we switch from coframe and density to a spinor field. All our
analysis, including Theorem~\ref{main theorem}, carries over to the
case of curved spacetime. We chose our (1+2)-dimensional spacetime
to be flat only to make the exposition clearer.

\subsection{Exclusion of gravity}

We assumed the (1+2)-dimensional metric $g$ to be
prescribed (fixed) and the coframe $\vartheta$ to be chosen so as to
satisfy the kinematic constraint (\ref{constraint for coframe}). As
explained in subsection \ref{Curved spacetime}, the fact
that we chose the metric $g$ to be Minkowski is irrelevant and all
our analysis carries over to the case of an arbitrary Lorentzian
metric in dimension 1+2. The important thing is that the metric $g$
is not treated as a dynamical variable. This means that we chose to
exclude gravity from our model.

On the other hand, in teleparallelism it is traditional to view the
metric as a dynamical variable. In other words, in teleparallelism
it is customary to view~(\ref{constraint for coframe}) not as a
kinematic constraint but as a definition of the metric and,
consequently, to vary the coframe $\vartheta$ without any
constraints. This is not surprising as most, if not all, authors who
contributed to teleparallelism came to the subject from General
Relativity.

It appears that the idea of working with a coframe subject to the
kinematic constraint (\ref{constraint for coframe}) is new.

\subsection{Our choice of Lagrangian}

We chose a very particular Lagrangian density~(\ref{basic Lagrangian density})
containing only one irreducible piece of torsion (axial) whereas in
teleparallelism it is traditional to choose a more general
Lagrangian containing all three pieces
(axial, vector and tensor)
of the torsion tensor
\begin{equation}
\label{definition of torsion extended}
\mathbf{T}^\mathrm{ax}:=\mathbf{o}_{\mathbf{j}\mathbf{k}}
{\bm\vartheta}^{\mathbf{j}}\otimes d{\bm\vartheta}^{\mathbf{k}},
\end{equation}
see formula (26)
in \cite{MR2419830}.
Note that when Einstein introduced teleparallelism
\cite{1001.4726} he neglected the axial piece
(\ref{definition of axial torsion extended})
completely.

In choosing our particular Lagrangian density (\ref{basic Lagrangian
density}) we were guided by the principles of conformal invariance,
simplicity and analogy with Maxwell's theory. The analogy with
Maxwell's theory is that we characterize the field strength by a
differential form, replacing the electromagnetic tensor (2-form) by
axial torsion (3-form). It appears that the
Lagrangian density~(\ref{basic Lagrangian density}) was never examined.

\subsection{Density as a dynamical variable}

We took the positive density of our continuum to be a dynamical
variable whereas in teleparallelism the tradition is to prescribe it
as $\rho=\sqrt{|\det g|}\,$. Taking $\rho$ to be a dynamical
variable is, of course, equivalent to introducing an extra real
positive scalar field into our model. It appears that the idea of
making the density a dynamical variable is also new.

\subsection{Electron in dimension 1+3}

The major outstanding issue is whether we can reformulate the Dirac
equation in dimension 1+3 using our approach. This would mean
starting from (1+3)-dimensional spacetime, performing a
Kaluza--Klein extension to dimension 1+4, choosing the conformally
invariant Lagrangian density (\ref{basic Lagrangian density}) and so
on, as described in Section \ref{Our model}.

It seems that the equation we get starting from
(1+3)-dimensional spacetime and performing the construction
described in Section \ref{Our model} is not the Dirac equation in
dimension 1+3. Our analysis is heavily dependent on dimension and,
when starting from (1+3)-dimensional spacetime, we do not appear to
get a factorization of the Lagrangian density of the type
(\ref{factorization formula}).

However, the equation we get in dimension 1+3, although nonlinear,
seems to be very similar to the Dirac equation. The natural way of
testing how close our equation is to the Dirac equation would be to
calculate the energy spectrum of the electron in a given static
electromagnetic field, starting with the case of the Coulomb
potential (hydrogen atom).

\subsection{Similarity with the Ashtekar--Jacobson--Smolin construction}

The analysis presented in our paper exhibits certain similarities
with \cite{MR933459,MR1634502} in that a 3-dimen\-sional (or, in our
case, (1+2)-dimensional) coframe $\vartheta$ is used as a dynamical
variable and that a second order partial differential equation is
reduced to a first order equation.

\appendix

\section{Notation}
\label{Notation}

Our notation follows
\cite{MR2573111,1001.4726,prd2007}.
The only difference with
\cite{MR2573111,prd2007}
is that in the latter the Lorentzian metric
has opposite signature.
In \cite{1001.4726} the signature is the same as in the current paper,
i.e. the (1+3)-dimensional metric
has signature $\,{-+++}\,$.

We use Greek letters for tensor (holonomic) indices and Latin
letters for frame (anholonomic) indices.

We identify differential forms with covariant antisymmetric tensors.
Given a pair of real covariant antisymmetric tensors $P$ and $Q$ of
rank $r$ we define their dot product as
$
P\cdot Q:=\frac1{r!}P_{\alpha_1\ldots\alpha_r}Q_{\beta_1\ldots\beta_r}
g^{\alpha_1\beta_1}\ldots g^{\alpha_r\beta_r}
$.
We also define $\|P\|^2:=P\cdot P$.

We define the action of the Hodge star on a rank
$r$ antisymmetric tensor $R$ as
$
(*R)_{\alpha_{r+1}\ldots\alpha_3}:=(r!)^{-1}\,
R^{\alpha_1\ldots\alpha_r}\varepsilon_{\alpha_1\ldots\alpha_3}
$
where $\varepsilon$ is the totally antisymmetric quantity,
$\varepsilon_{012}:=+1$.

We use two-component complex-valued spinors (Weyl spinors) whose
indices run through the values $1,2$ or $\dot1,\dot2$. Complex
conjugation makes the undotted indices dotted and vice versa.

We define the ``metric spinor''
\begin{equation}
\label{metric spinor}
\epsilon_{ab}=\epsilon_{\dot a\dot b}=
\epsilon^{ab}=\epsilon^{\dot a\dot b}=
\begin{pmatrix}
0&-1\\
1&0
\end{pmatrix}
\end{equation}
and choose Pauli matrices
\begin{equation}
\label{Pauli matrix temporal}
\sigma_{0\dot ab}\!=\!
\begin{pmatrix}
1&0\\
0&1
\end{pmatrix}
\!=-\sigma^0{}_{\dot ab},
\end{equation}
\begin{equation}
\label{Pauli matrices spatial}
\sigma_{1\dot ab}\!=\!
\begin{pmatrix}
0&1\\
1&0
\end{pmatrix}
\!=\sigma^1{}_{\dot ab},
\quad
\sigma_{2\dot ab}\!=\!
\begin{pmatrix}
0&-i\\
i&0
\end{pmatrix}
\!=\sigma^2{}_{\dot ab},
\quad
\sigma_{3\dot ab}\!=\!
\begin{pmatrix}
1&0\\
0&-1
\end{pmatrix}
\!=\sigma^3{}_{\dot ab}\,.
\end{equation}
Here the first index enumerates rows and the second enumerates
columns

\section{Nonlinear second order equations
which reduce to pairs of linear first order equations}
\label{Nonlinear second order equations}

Let $\Omega$ be an open subset of $\mathbb{R}^n$. We work with
(infinitely) smooth vector functions $\Omega\to\mathbb{C}^m$
writing these as columns of $m$ complex scalars. In this appendix
``vector'' does not carry a differential geometric meaning
because we are not interested in coordinate transformations.
We use Cartesian coordinates $x^1,\ldots,x^n$.

Given a pair of vector functions $u,v:\Omega\to\mathbb{C}^m$ we
define their inner product in the standard Euclidean manner as
$
(u,v):=\int_\Omega v^*u\,dx^1\ldots dx^n
$
where the star~$*$ denotes Hermitian conjugation. This integral
need not converge as we will be using it only
for the purpose of defining the formal adjoint of a differential
operator, see next paragraph.

Let $A_\pm$ be a pair of formally self-adjoint (symmetric) first
order linear partial differential operators (differential
expressions) with smooth coefficients acting on smooth vector
functions $\Omega\to\mathbb{C}^m$. We do not introduce any
boundary conditions.

Put
\begin{equation}
\label{first order Lagrangian density}
L_\pm(u):=\operatorname{Re}(u^*A_\pm u).
\end{equation}
It is easy to see that $L_\pm(u)$ is the Lagrangian density for the
partial differential equation $A_\pm u=0$. Namely, if one writes
down the action (variational functional)
$
S_\pm(u):=\int_\Omega L_\pm(u)\,dx^1\ldots dx^n
$
then the corresponding field equation (Euler--Lagrange equation) is
$A_\pm u=0$.

Let us now define a new Lagrangian density
\begin{equation}
\label{second order Lagrangian density}
L(u):=
\frac{L_+(u)\,L_-(u)}{L_+(u)-L_-(u)}
\end{equation}
and corresponding action
$
S(u):=\int_\Omega L(u)\,dx^1\ldots dx^n
$.
The field equation for the Lagrangian density
(\ref{second order Lagrangian density}) is, of course,
second order and nonlinear.

Note that the notation in this appendix is self-contained and the
Lagrangian densities (\ref{first order Lagrangian density}),
(\ref{second order Lagrangian density}) should not be confused
with the Lagrangian densities
(\ref{our Lagrangian density with A via spinor for eta}),
(\ref{Dirac Lagrangian density}) introduced in the main text
(the latter have an extra subscript).

The main result of this appendix is

\begin{lemma}
\label{Nonlinear second order equations lemma}
Let $u:\Omega\to\mathbb{C}^m$ be a vector function satisfying
the condition
\begin{equation}
\label{denominator is nonzero}
L_+(u)\ne L_-(u).
\end{equation}
Then $u$
is a solution of the field equation
for the Lagrangian density $L$
if and only if it is a solution of
the equation $A_+u=0$
or the equation $A_-u=0$.
\end{lemma}

\begin{proof}
The explicit formula for the operator $A_\pm$ is
\begin{equation}
\label{explicit formula for operator A plus minus}
A_\pm=iB_\pm^\alpha\partial_\alpha+\frac i2(\partial_\alpha B_\pm^\alpha)
+C_\pm
\end{equation}
where $B_\pm^\alpha$ and $C_\pm$ are some smooth Hermitian $m\times m$
matrix functions and the index $\alpha$ runs through the values $1,\ldots,n$.
Substituting (\ref{explicit formula for operator A plus minus}) into
(\ref{first order Lagrangian density}) we get
\begin{equation}
\label{explicit formula for first order Lagrangian density}
L_\pm(u)=
\frac i2\bigl[
u^*B_\pm^\alpha\partial_\alpha u
-
(\partial_\alpha u^*)B_\pm^\alpha u
\bigr]
+u^*C_\pm u.
\end{equation}

Now take an arbitrary smooth function $h:\Omega\to\mathbb{R}$.
Examination of formula
(\ref{explicit formula for first order Lagrangian density}) shows that
\begin{equation}
\label{scaling covariance}
L_\pm(e^h u)=e^{2h}L_\pm(u).
\end{equation}
We call the property (\ref{scaling covariance}) \emph{scaling
covariance}. Scaling covariance is a remarkable feature of the
Lagrangian density of a formally self-adjoint first order linear partial
differential operator.

Formulas (\ref{second order Lagrangian density}) and
(\ref{scaling covariance}) imply that the Lagrangian density $L$
also possesses the property of scalar covariance, i.e.
$L(e^h u)=e^{2h}L(u)$ for any smooth $h:\Omega\to\mathbb{R}$.
Thus, all three of our Lagrangian densities,
$L$, $L_+$ and $L_-$, have this property.

Observe now that if the vector function $u$ is a solution
of the field equation for
some Lagrangian density $\mathcal{L}\,$
possessing the property of scaling covariance
then $\mathcal{L}(u)=0$. Indeed, let us perform a scaling
variation of our vector function
\begin{equation}
\label{scaling variation}
u\mapsto u+\delta u= u+hu=e^hu+O(h^2)
\end{equation}
where $h:\Omega\to\mathbb{R}$
is an arbitrary ``small'' smooth function
with compact support, $h\in C_0^\infty(\Omega;\mathbb{R})$. Then
$0=\delta\!\int\!\mathcal{L}(u)=2\int h\mathcal{L}(u)$
which holds for arbitrary $h$ only if $\mathcal{L}(u)=0$.

In the remainder of the proof the variation
$\delta u:\Omega\to\mathbb{C}^m$ of the vector
function $u:\Omega\to\mathbb{C}^m$ is arbitrary
and not necessarily of the scaling type
(\ref{scaling variation}). The only assumption
is that $\delta u\in C_0^\infty(\Omega;\mathbb{C}^m)$.

Suppose that $u$ is a solution of the field equation for
the Lagrangian density $L_+$.
[The case when $u$ is a solution of the field equation for
the Lagrangian density $L_-$ is handled similarly.]
Then $L_+(u)=0$ and, in view of formula (\ref{denominator is nonzero}),
$L_-(u)\ne0$.
Varying $u$ we get
\begin{multline*}
\delta\!\int\!L(u)
=
\int
\frac{L_-(u)}
{L_+(u)-L_-(u)}
\,\delta L_+(u)
+
\int
L_+(u)
\,\delta\frac{L_-(u)}
{L_+(u)-L_-(u)}
\\
=-\int\delta L_+(u)
=-\delta\!\int\!L_+(u)
\end{multline*}
so
\begin{equation}
\label{formula for variation of our action}
\delta\!\int\!L(u)=-\delta\!\int\!L_+(u)\,.
\end{equation}
We assumed that $u$ is a solution of the field equation for
the Lagrangian density $L_+$ so
$\delta\!\int\!L_+(u)=0$ and
formula~(\ref{formula for variation of our action}) implies that
$\delta\!\int\!L(u)=0$. As the latter is true for an
arbitrary variation of $u$ this means that
$u$ is a solution of the field equation for the Lagrangian
density~$L$.

Suppose that $u$ is a solution of the field equation for
the Lagrangian density $L$.
Then $L(u)=0$ and formula~(\ref{second order Lagrangian density})
implies that either $L_+(u)=0$ or $L_-(u)=0$;
note that in view of (\ref{denominator is nonzero}) we cannot have simultaneously
$L_+(u)=0$ and $L_-(u)=0$.
Assume for definiteness that $L_+(u)=0$.
[The case when $L_-(u)=0$ is handled similarly.]
Varying~$u$ and repeating the argument from the previous paragraph
we arrive at (\ref{formula for variation of our action}).
We assumed that $u$ is a solution of the field equation for
the Lagrangian density $L$ so
$\delta\!\int\!L(u)=0$ and formula
(\ref{formula for variation of our action}) implies that
$\delta\!\int\!L_+(u)=0$. As the latter is true for an
arbitrary variation of $u$ this means that
$u$ is a solution of the field equation for the Lagrangian
density~$L_+$.~$\square$
\end{proof}

\begin{remark}
It may seem that the variational proof presented above is
``insufficiently rigorous''. An alternative ``completely rigorous''
way of proving Lemma~\ref{Nonlinear second order equations lemma}
is to write down the field equation for the Lagrangian density
(\ref{second order Lagrangian density}),
(\ref{explicit formula for first order Lagrangian density})
explicitly and analyze this second order nonlinear partial differential equation.
The result, of course, remains the same, only the calculations become
much longer.
\end{remark}

\begin{remark}
Examination of the proof of
Lemma~\ref{Nonlinear second order equations lemma}
shows that the fact that the differential operators
$A_\pm$ are linear and first order is not important.
What is important is that their Lagrangian densities possess
the scaling covariance property (\ref{scaling covariance}).
As the Lagrangian density (\ref{second order Lagrangian density})
possesses this property as well, our construction
admits an obvious extension which gives a hierarchy of non\-linear partial
differential equations
which reduce to several separate equations.
\end{remark}

\begin{example}
Let us give an elementary example illustrating the use of
Lemma~\ref{Nonlinear second order equations lemma}.
Consider the pair of linear first order ordinary differential equations
\begin{equation}
\label{example equation 1}
iu'\pm u=0
\end{equation}
where $u:\mathbb{R}\to\mathbb{C}$ is a scalar function.
Let us write down the corresponding Lagrangian densities
$L_\pm(u)=\frac i2(\bar uu'-u\bar u')\pm|u|^2$ in accordance with
formula~(\ref{first order Lagrangian density})
and form a new Lagrangian density
$
-2L(u)
=\bigl(
\frac{\bar uu'-u\bar u'}{2|u|}
\bigr)^2\!+|u|^2
$
in accordance with formula (\ref{second order Lagrangian density}).
The latter gives the field equation (Euler--Lagrange equation)
\begin{equation}
\label{example equation 2}
\left(
\frac{\bar uu'-u\bar u'}{2|u|^2}\,u
\right)'
+\frac{(\bar uu')^2-(u\bar u')^2}{4|u|^4}\,u
+u=0.
\end{equation}
Lemma~\ref{Nonlinear second order equations lemma} tells us that a smooth
nonvanishing function $u$ is a solution of
equation~(\ref{example equation 2})
if and only if it is a solution of one of the two
equations~(\ref{example equation 1}).
Of course, this fact can be checked directly by switching to
the polar representation
$u=re^{-i\varphi}$
where
$r:\mathbb{R}\to(0,+\infty)$ and
$\varphi:\mathbb{R}\to\mathbb{R}$.
\end{example}

\end{document}